\documentclass{article}
\usepackage{fullpage,url,amsmath,amsfonts,amssymb, amsthm,mathtools}
\newcommand{\transpose}{\text{${}^{\text{T}}$}}
\newcommand{\floor}[1]{\lfloor #1 \rfloor}
\newcommand{\T}{\transpose}
\newcommand{\C}{\mathcal{C}}
\newcommand{\D}{\mathcal{D}}
\newcommand{\ti}{\times}

\newcommand{\al}{\alpha}
\newcommand{\be}{\beta}
\newcommand{\de}{\delta}

\newcommand{\om}{\omega}

\newcommand{\f}{\frac}
\newcommand{\di}{\displaystyle}
\newenvironment{smatrix}{\bigl(\begin{smallmatrix}}{\end{smallmatrix}\bigr)}
\newenvironment{ssmatrix}{\left(\begin{smallmatrix}}{\end{smallmatrix}\right)}

\newtheorem{thm}{Theorem}[section]
\newtheorem{proposition}[thm]{Proposition}

\newtheorem{lemma}[thm]{Lemma}

\newtheorem{alg}{\bf{Algorithm}}
\newtheorem{example}[thm]{Example}

\newcommand{\Z}{\mathbb{Z}}

\newcommand{\F}{\mathcal{F}}

\title{Ultimate linear block and convolutional codes\footnote{ AMS Classification 2020: 94B05, 94B10, 15B99}}
\author{Ted Hurley\footnote{University of Galway (previously: National University of Ireland Galway) \\ email: Ted.Hurley@universityofgalway.ie}}
\date{}

\begin{document}
\maketitle
\begin{abstract} 

 Codes considered as structures within unit schemes greatly extends the availability of linear block and convolutional  codes  and allows the construction of these codes to required length, rate, distance and type. Properties of a code emanate from properties of the unit from which it was derived. 
 Orthogonal units, units in group rings, Fourier/Vandermonde units and related units are used to construct and analyse linear block and convolutional codes and to construct these to predefined  length, rate, distance and type. Self-dual, dual containing, quantum error-correcting and linear complementary dual codes are constructed for both linear block and convolutional codes.  
 Low density parity check  linear block and convolutional codes 
 are constructed  with no short cycles in the control matrix. 
  \end{abstract}
\section{Introduction} 
Unit-derived methods for coding theory were initiated in  \cite{hur1,hur0} and are further developed in \cite{unitderived,hurleyquantum,hurleyunit}; an introduction in  chapter form is available in \cite{hur2,hur11}. All linear block codes over fields are unit-derived and the method is used here  
in a number of different directions to devise new series of linear block codes to required length, rate, distance and type. These are also extended to establishes methods for  creating and analysing convolutional codes and to build series of these to required type, length, rate and  field type. Methods derived in \cite{hurleyunit} for establishing infinite series of codes with rates approaching a given rational $R$, $(0<R<1)$, and relative distances approaching $(1-R)$ are applicable.  
Some of the codes derived have applications in solving {\em underdetermined systems of equations}, 
see \cite{hurleyunder}.

The (free) distance of a convolutional code  derived may often be determined from the unit used and in general is of an order better than the distance of a  linear block code of the same rate devised. 
Efficient decoding methods are available.  

All linear block codes over fields are  unit-derived codes, see Proposition \ref{unit} below, although they may not have been derived from this outlook initially. Thus all linear block codes over fields may be devolved from within unit schemes. By looking at a unit scheme in general and selecting constituents, convolutional codes  to required type and distance are constructed and analysed.  

Codes both linear and convolutional are designed from unit schemes. 
Many known codes are not designed in this way initially. However by looking at these codes as structures from units further information and further code designs are obtainable; for example a convolutional code of the same rate may be designed from that unit which performs better and has order distance twice that of the original linear code. An illuminating examples of this is given in section \ref{hamm} below: 
  A unit-derived form of the Hamming $[7,4,3]$ code is derived and  used to construct  a memory $1$ convolutional (binary) code of distance $6$ and rate  $\f{4}{7}$; the convolution  code formed is of type $(7,4,4;1,6)$ --  see section \ref{conv} below for relevant convolutional code parameter definitions. Another example, \ref{unit1} below, is where the unit part of the Golay binary $[24,12,8]$ code is used to design convolutional memory $3$ codes with rates  $\f{3}{4}$ and $\f{1}{4}$; the larger rate one  is dual containing and the smaller rate one has distance $18$. 
 Hamming, Golay codes may thus be extended respectively to `convolutional Hamming' and `convolutional Golay' codes. 
 The process is of a general nature and such treatment may be given to other known codes. 
 

The classes of linear block and convolutional codes available are greatly expanded. McEliece (\cite{mceliece}) remarks: `` A most striking fact is the lack of algebraic constructions of families of convolutional codes.''; Blahut, \cite{blahut}, p.312, writes: ``No general method is yet known for constructing a family of high performance multiple-error-correcting code..The codes now in common use have been found by computer search''. Computer search is no longer practical. Multiple  algebraic methods for designing series of convolutional codes are now available. 
and such codes to required length, rate and type required are constructed by the unit-derived and related methods as applied here and previously in  \cite{hurleyunit,hur111,hur10}. Convolutional codes galore are available from the unit-derived methods.
The coding structure lies within the unit structure about which much is known. 
In mathematics we often think of breaking a structure into more manageable parts but here we think of looking at the bigger structure within which required embedded structures exist. 


Using special types of units such as  orthogonal units, Vandermonde/Fourier units or  units in group rings allows the construction of special types of codes such as {\em self-dual codes}, {\em dual containing codes}, {\em complementary dual codes} and {\em quantum code} as well as codes to specified lengths, rates, distances and over specified finite fields. Codes over particular required fields, such as over fields of characteristic $2$ or codes over prime fields (for which modular arithmetic is available) and series of such, are constructible by the methods. Special linear block and convolutional codes may also be induced from Hadamard matrices in general from their unit form and this is dealt with in detail in \cite{induced}. In section \ref{hadamard} some samples are given as an introduction to what is achievable.

Looking at units in group rings allows the construction of low density parity check (LDPC) linear block and LDPC convolutional codes and these are constructed  with no short cycles in the check  matrix.

 Dual-containing codes have their own intrinsic interest but are also used for designing {\em quantum error-correcting codes} by the CSS method, \cite{calderbank,calderbank1,steane}. Here then {\em convolutional quantum error correcting} codes of  different lengths and rates  are  explicitly constructed. 
  Linear complementary dual, LCD, codes have been studied extensively in the literature. For
background, history and general theory 
on LCD codes,
consult the articles \cite{sihem3,sihem2,sihem4, sihem} by
Carlet, Mesnager, Tang, Qi and Pelikaan. LCD codes were originally
introduced by Massey in \cite{massey,massey2}.  These codes have been
studied amongst other things for improving the security of information on sensitive devices
against {\em side-channel attacks} (SCA) and {\em fault non-invasive
  attacks}, see \cite{carlet}, and have found use in {\em data
  storage} and {\em communications' systems}.  

The relationships between DC linear block codes and LCD convolutional codes  and between LCD linear block codes and convolutional DC codes when formed from the same unit scheme are quite remarkable.


Hermitian codes over fields of the form $GF(q^2)$ may be  by looking at {\em unitary matrices}; this  was initiated in \cite{hurleyunit} but  full development is left  to later work. 

Requiring one of $\{U,V\}$  in $UV=I$ to be of low density enables the construction of low density parity check (LDPC) linear block and convolutional LDPC codes by the unit-derived method. These are constructed so that there are no short cycles in the control matrix using units in group rings. 
The linear block case for such LDPC codes has been  dealt with in \cite{hurley33} and the convolutional case follows from the unit-derived method. 
Iterative decoders for low density parity check  codes are impacted by short cycles. Here for a given unit scheme, described in a precise way,  multiple such codes, all with no short cycles, are  constructed and with prescribed rate and dimension.  
 These  LDPC  codes have many applications and can further be stored using an algebraic  `short' formula which for example is important in applications requiring low storage and low power. 


Some of this  work is  additional and complementary to that in \cite{hur111,hur10}. 
The paper \cite{hur111} has appeared on ArXiv only, having been rejected elsewhere. 
As pointed out in \cite{hur111},  convolutional codes which appeared previously in the literature are very special cases of constructions using these types of methods.  


The unit-derived method allows the construction of multiple linear block codes and multiple convolutional codes from the same unit. 
`Manufacturing' of different and sophisticated `models' is made relatively easy. 

 A number of examples are given here which of necessity are of relatively small length and these can be looked upon as prototype examples for large length constructions which are attainable. The codes are easily implementable once the units from which they are derived are formed. The brilliant Computer Algebra system GAP, \cite{gap}, proves extremely useful in constructing examples and verifying  distances. 
Coding theory background is contained in section \ref{conv}.
\subsection{Background on  coding theory}\label{conv}
Basics on linear block coding theory may be found in any of \cite{blahut,joh,mceliece,macslo} and many others. The notation $[n,r,d]$ is used here for a linear block code of length $n$, dimension $r$, and (minimum) distance $d$. The rate is then $\f{r}{n}$. 
A maximum distance separable (mds) linear block code is one of the form $[n,r,n-r+1]$ where the maximum distance possible for the length and rate is achieved. 

Different equivalent definitions for convolutional codes are given in the literature. The notation and definitions used here follow that given in \cite{ros,smar,rosenthal1}. 
A rate $\frac{k}{n}$ convolutional code with parameters $(n,k,\de)$ over a field $\F$ is 
 a submodule of $\F[z]^n$ generated by a reduced
basic matrix $G[z] =(g_{ij}) \in \F[z]^{r\ti n}$ of rank $r$ where $n$ is
the length,  $\de = \sum_{i=1}^r \de_i$ is the {\em degree}
with  $\de_i= \max_{1\leq j\leq r}{\deg g_{ij}}$. Also 
$\mu=\max_{1\leq i\leq r}{\de_i}$ is known as the {\em memory} of the
code and then the code may then be given with parameters $(n,k,\de;\mu)$. 
The parameters $(n, r,\delta;\mu, d_f)$ are  used for such a code 
with free (minimum) distance $d_f$.

 Suppose $\C$ is a convolutional code in $\F[z]^n$ of rank $k$. A generating matrix $G[z] \in
\F[z]_{k\times n}$ of $\C$ having rank $k$ is called a
{\em generator} or {\em encoder matrix}  of $\C$. 
A matrix $H \in \F[z]_{n\times(n-k)}$ satisfying $\C = \ker H =
\{v \in \F[z]^n : vH = 0 \}$ is said to be a {\em control matrix} or
    {\em check matrix} of the code $\C$. 

 Convolutional codes can be {\em
   catastrophic or non-catastrophic}; see
 for example \cite{mceliece} for the basic definitions. A
catastrophic convolutional code is prone to catastrophic error
propagation and is not much use. A
convolutional code described by a generator matrix  with {\em right
polynomial inverse} is a non-catastrophic code; this is sufficient for
our purposes. The designs given here for  the generator matrices allow for
specifying directly the control matrices and the right polynomial
inverses where appropriate. 
There exist very few algebraic constructions for designing convolutional codes and search methods limit their size
and availability, see McEliece \cite{mceliece} for discussion and also \cite{alm1,alm2,guar,mun}.

By Rosenthal and Smarandache, \cite{ros},
 the maximum free distance attainable by an
$(n,r,\delta)$ convolutional  code is  $(n-r)(\floor{\frac{\de}{r}}+1)+
\de +1$. The case  $\delta
=0$, which is the case of zero memory, corresponds to the linear 
Singleton bound $(n-r+1)$. 
The bound $(n-r)(\floor{\frac{\de}{r}}+1)+
\de +1$  
is then called  the {\em generalised Singleton  bound}, \cite{ros}, GSB,
and a convolutional code attaining this bound is known as  an {\em mds 
convolutional code}. The papers \cite{ros} and \cite{smar} 
are major beautiful contributions to this area.

The criteria for a convolutional code to be an mds code are given in terms of the parameters for a convolutional code and the criteria for a linear block code to be an mds code are given in terms of the parameters for a linear block code. 

Let $G(z)$ be the generator matrix for a convolutional code $\C$ with memory $m$.
Suppose $G(z)H\T(z) = 0$, so that $H\T(z)$ is a control matrix, and then $H(z^{-1})z^m$  generates the {\em convolutional dual code} of $\C$, see \cite{dual2} and \cite{dualconv}. This is also known as the {\em module-theoretic dual code}.\footnote{In convolutional coding theory, the idea of {\em dual code} has 
two meanings. 
The other dual convolutional code defined  
  is called {\em the sequence space dual}; the generator
  matrices for these two types are related by a specific formula.}
The code is then dual-containing provided the code generated by $H(z^{-1})z^m$ is contained in the code generated by $G(z)$.

Let $G(z)$ be the generator matrix for a convolutional $(n,r)$ code $\C$. Code words will consist of $P(z)G(z)$ where $P(z)$ is a polynomial in $z$ with coefficients which are $1\ti r$ vectors. The polynomial $P(z)$ is said to be an {\em information vector} for the code $\C$. The support of $P(z)$ is the number of non-zero coefficient vectors appearing in its expression as a polynomial.



\subsection{Dual-containing and linear complementary dual codes} The dual of a code $\C$ is denoted by $\C^{\perp}$. Note the definition of {\em dual code of a convolutional code} as given in subsection \ref{conv} above. A code $\C$ is said to be {\em dual containing,} written DC, if it contains its dual $\C^{\perp}$. 
Say a code is a {\em linear complementary dual}, written LCD, code provided it has trivial intersection with its dual.

\noindent Thus \\
$\C$ is a dual containing (DC) code $\Longleftrightarrow$ $\C \cap C^{\perp}=\C^{\perp}$ \\
$\C$ is a linear complementary dual (LCD) code $\Longleftrightarrow$ $\C \cap C^{\perp}=0$ \\
A {\em self-dual code} is a code $\C$  with $\C^{\perp}= \C$; this is an important  type of DC code. 

Constructions of convolutional DC and LCD codes were initiated in \cite{hurleyunit}. 
 DC convolutional codes are theoretically  interesting in themselves but in addition a  DC convolutional code  enables the construction of a {\bf convolutional quantum error-correcting  code} by the CSS method. 

LCD codes and DC codes are `supplemental' to
one another: $\C$ is DC if and only if 
$\C\cap\C^\perp = \C^\perp$ and $\C$ is LCD if and only if $\C \cap \C^\perp = 0$.
As noted in \cite{hurleyunit}, mds DC block linear codes lead to  the construction of mds LCD
convolutional codes and LCD mds  block linear codes leads to the construction of mds DC 
convolutional codes.

{\em Abbreviations used:}

\noindent DC: dual-containing

\noindent mds: maximum distance separable \footnote{Has different parameter requirements for linear block codes, convolutional codes and quantum codes.}

\noindent LCD: linear complementary dual

\noindent QECC: quantum error-correcting code

\noindent LDPC: low density parity check

\subsection{Designs achievable}\label{design} Propositions \ref{unit} to \ref{three}, which enable  the code constructions with properties,  are given explicitly and proven in section  \ref{unitderived}. 
The designs  that follow a proposition are given below. 
\begin{enumerate}
\item Proposition \ref{unit}:
Consequence: Design linear block codes from units.


  
\item  Proposition \ref{orthogonal}: Consequence: {\em Use Orthogonal matrices to design  LCD codes.} 

\item Proposition \ref{jut} (binary case) and Proposition \ref{orthogonalsq}:
  Algorithm: {\em Design self-dual codes from orthogonal type  units}. 
Noted: All self-dual linear block codes can be constructed in this way.  
  

  

\item 

  Proposition \ref{fourier} and Proposition \ref{fourier1}: Algorithm: 
 Design   {\em linear block mds, DC and quantum codes from Fourier/Vandermonde matrices.}  

   These are implicit in the paper \cite{unitderived}.

    \item  Proposition \ref{lcd}: Algorithm therefrom: Design LCD, mds codes from Fourier/Vandermonde type matrices: 




  \item  Proposition \ref{conv1}: Algorithm: Design  length $2n$ rate $\f{1}{2}$, memory $1$, convolutional code and describe the dual code.  





      \item  Proposition \ref{conv2}: Algorithm: Design convolutional self-dual codes leading to quantum convolutional code construction: 

        




  
            \item Proposition \ref{conv3}:  Algorithm:  Design  DC  convolutional codes; design quantum convolutional codes therefrom.  





            \item   
              Proposition \ref{conv4}:
      Algorithm therefrom:   Design DC convolutional codes of rate $> \f{1}{2}$ from orthogonal matrices and orthogonal like matrices; design convolutional quantum codes therefrom.  





\item\label{memory}  Propositions \ref{mem} and \ref{three}. Algorithm: Design higher memory convolutional DC codes,  quantum convolutional codes and LCD convolutional codes.   

\item Section \ref{ldpc}, Algorithm \ref{ldpcalg}. Design  LDPC linear block codes and LDPC convolutional codes  by applying unit-derived techniques to special units in  group rings; these can be designed with no short cycles in a control matrix. 
 
           \end{enumerate}

      
\subsection{Decoding} 
Efficient decoding techniques for unit-derived linear block codes from Fourier/Vandermonde type matrices are established in \cite{unitderived}, Algorithms 6.1, 6.2 and 6.3. The constructions quickly lead to the establishment of  error-correcting pairs for the codes; error-correcting pairs are due to Pellikan, \cite{pell}.  The algorithms are especially  useful for  solving underdetermined systems using error-correcting codes - see  \cite{hurleyunder}. 
Several algorithms exist for decoding convolutional codes, the most common ones being the Viterbi algorithm and the sequential decoding algorithm. Other types of decoding such syndrome decoding are also available.


\section{Unit-derived}\label{unitderived}

The unit-derived method for constructing and analysing linear block codes  was initiated in \cite{hur1,hur2,hur11,hur0} and continued in \cite{unitderived,hurleyquantum} and elsewhere.
Any linear block code can be derived by the unit-derived method although this may not have been the original line of thought in the construction of the code.
The unit-derived method gives further information on the code in addition to describing the  generator and control matrices. See example  \ref{hamm} below which uses the Hamming $[7,4,3]$ in its unit-derived code form. 

A linear block code with generator matrix $G$ and check matrix $H$ is described by
$GH\T= 0$. The matrix $H$ generates the dual of the code and  the term `control matrix' is also be used for the matrix $H\T$. 
The basic unit-derived method is obtained as  follows:
$U$ is an invertible $n\ti n$  matrix and is broken up as $U=\begin{ssmatrix} A \\ B \end{ssmatrix}$. The inverse of $U$ has a compatible form $\begin{ssmatrix}  C & D \end{ssmatrix}$ so that  
$ \begin{pmatrix} A \\ B \end{pmatrix} \begin{pmatrix}  C & D \end{pmatrix}=I_n$.

Now $A$ has size $r\ti n$ for some $r$ and then $B$ has size $(n-r)\ti n$, $C$ has size $n\ti r$ and $D$ has size $n\ti (n-r)$. 
Then precisely $AC=I_{r}, AD=0_{r\ti (n-r)}, BC = 0_{(n-r)\ti r}, BD = I_{(n-r)}$.
So $AD=0$ defines an $[n,r]$ code $\C$ where $A$ is the generator matrix, $D$ is a control matrix and $D\T$ generates the dual code of $\C$.

A more general form of the unit-derived method, see \cite{hur1,hur2,hur11,hur0},
 is as follows:  Given a unit matrix system $UV= I_n$, taking any $r$ rows of $U$ gives an $[n,r]$ code and a control matrix is obtained by eliminating the corresponding columns of $V$. Thus many codes may be derived, and codes of a particular type, from a single unit scheme.

Properties of the units are used to obtain  properties of the codes, and units are formed with a particular type, length, rate, distance or field type in mind. Infinite series of required codes are also constructed and analysed; see for example the paper \cite{hurleyorder}.
The number of choices of $r$ rows from $n$ is ${n\choose r}$ , thus deriving  many codes from a single unit scheme. Having a big choice is also useful in producing cryptographic schemes from large unit schemes.

In the basic unit scheme above $A$ is the generator matrix of an $[n,r]$ code $\C$ and then $D$, which is a $(n-r)\ti r$ matrix, is a check matrix for $\C$. Every code over a field can be given in this unit-derived form, Proposition \ref{unit}. But note that $B,C$ have been ignored! They can also be used to describe a `complementary code' but even better can be used to form a convolutional code with $A,D$. Distances of the convolutional codes formed from a unit can often be determined in terms of a sum of the distances of linear codes formed from that unit.  Convolutional codes have in addition their own efficient decoding algorithms, such as Viterbi algorithm and sequential decoding algorithm.

The matrices $\{A,B,C,D\}$ have full ranks as they are parts of invertible matrices.

Every linear block code over a field arises as  a unit-derived code. 

\begin{proposition}\label{unit}  Let $\C$ be a linear code over a field. Then $\C$ is equivalent to a unit-derived code.
\end{proposition}

 \begin{proof}
  Assume $\C$ is an $[n,r]$ code with generator matrix $A$ and check matrix $H$. Then $AH\T=0$ for an $r\ti n$ matrix $A$, and an  $(n-r)\ti n$ matrix $H$; here  $0=0_{r\ti (n-r)}$. Let $\{e_1,e_2,\ldots, e_r\}$ be the rows of $A$ which are linearly independent. Extend these to a basis $\{e_1,e_2,\ldots, e_r, e_{r+1},\ldots, e_n\}$ for the whole space $n$-dimensional space. Let $B=\begin{ssmatrix}e_{r+1} \\ \vdots \\e_n \end{ssmatrix}$ and $G=\begin{pmatrix}A \\B \end{pmatrix}$. Then $G$ is invertible with inverse given by $K=G^{-1}=
  \begin{pmatrix} C &D \end{pmatrix}$ where $C$ is an $n\ti r$ matrix and $D$ is an $n\ti (n-r)$ matrix. Thus $ \begin{pmatrix} A \\ B \end{pmatrix} ( C  D) =I_n$. Then $AD=0$. Now $D\T$ has rank $(n-r)$ and $H\T$ has rank $(n-r)$ and hence the code generated by $A$ with check matrix $H$ is equivalent to the code generated by $A$ with check matrix $D\T$. ($D$ and $H\T$ generate the null space of $Ax=0$ and have the same rank.)

 \end{proof}

 A code may not be originally  constructed as a unit-derived code but it is useful to look at a code in this manner which  leads to further and better  constructions and in particular to  constructions of convolutional codes.
   A code is a  structures which is part of a bigger structure on which more is already known. Using the  
bigger structure to construct and analyse the embedded structures has many advantages. Multiple codes of a particular type may be deduced from just one unit.

 \begin{example}\label{hamm1} In section \ref{hamm} the Hamming $[7,4,3]$ is given  as a unit-derived code and from this a Hamming-type convolutional $(7,4,3;1,6)$ (binary) code is constructed; the distance is twice that of the Hamming code. Decoding techniques for convolutional codes can be employed. 
 \end{example}
 
  If we require particular types of codes as for example  DC (including self-dual) codes or LCD codes, then look for particular types of units which give such codes in the unit-derived way.


The convolutional codes  derived in \cite{hurleyunit} use  unit-derived methods from Vandermonde/Fourier and other well-behaved  matrices. 

Unit-derived codes may also be obtained from a scheme where  $UV=\al I_n, \al \neq 0$. The process is similar: Choose any $r$ rows of $U$ for a generator matrix and a check matrix is obtained by eliminating the corresponding columns of $V$. This is useful when considering Vandermonde/Fourier matrices, Propositions \ref{fourier} and \ref{fourier1},  and Hadamard-type matrices, section \ref{hadamard}. 
\subsection{Using orthogonal units}\label{secorthog}
\begin{proposition}\label{orthogonal} Let $U$ be an orthogonal matrix. Then any unit-derived block linear code from $U$ is an LCD (linear complementary dual) code.
  
    \end{proposition}

\begin{proof} Now $UU\T = I_n$. Thus the unit scheme is 
  $UU\T= \begin{smatrix} A \\ B \end{smatrix}  \begin{smatrix}  C & D \end{smatrix}=I_n$ for matrices $A,B,C,D$ where $A$ is of size $r\ti n$, $B$ is os size $(n-r)\ti n$, $C$ is of size $n\ti r$ and $D$ is of size $n\ti (n-r)$. Denote the code generated by $A$ by $\C$. This code has control matrix $D$, which means   $AD=0$. Now $U\T =  \begin{smatrix}  C & D \end{smatrix}$  and so $U=  \begin{smatrix} C\T \\ B\T \end{smatrix}$ giving that $C\T=A, D\T= B$. Thus $AD=0$ is the same as $AB\T =0$.

  Hence $B$ generates the dual code of $\C$. Now no non-trivial sum of rows of $A$ can be a sum of rows of $B$ as $ U=\begin{smatrix} A \\ B \end{smatrix}$ is non-singular.
Hence $\C\cap \C^{\perp}=0$ as required.
  \end{proof}
The following Proposition is shown in a similar manner to Proposition \ref{orthogonal}. 
\begin{proposition}\label{orthogonal1}
  Let $X$ be an $n\ti n$ matrix such that $XX\T=\al I_n$ for $\al \neq 0$. Suppose $X$ is broken as follow: $X= \begin{ssmatrix} A \\ B\end{ssmatrix}$ and this gives $\begin{ssmatrix}A \\ B\end{ssmatrix} \begin{ssmatrix} A\T & B\T \end{ssmatrix} = \al I_n$ where $A$ has size $r\ti n$. Then $A$ generates an $[n,r]$ LCD code and $B$ generates the dual of this code.
      \end{proposition}

Orthogonal matrices are thus a rich source for LCD codes. Given an orthogonal $n\ti n$ matrix $U$ any $r$  rows of $U$ may be chosen as the generator matrix for a $[n,r]$ code and this code is then an LCD linear block code.

An orthogonal matrix may also be used to form a self-dual code by combining it with  an identity as described in the following Propositions \ref{jut} and \ref{orthogonalsq}. 
The extended Hamming $[8,4,4]$ and Golay $[24,12,8]$ codes are constructed in this way, see examples \ref{self} and \ref{self1} below. Other self-dual codes may be constructed in a similar manner from orthogonal matrices. 

The following Proposition  is known but is given here in a form suitable for the constructions. 
\begin{proposition}\label{jut} Let $X$ be an orthogonal $n\ti n$ matrix in a field of characteristic $2$. Then the matrix $A=(I_n,X)$ generates  a self-dual $[2n,n]$ matrix. Conversely if $A=(I_n,X)$ is a self-dual code where $X$ is an $n\ti n$ matrix in a field of characteristic $2$, then $X$ is orthogonal. 
\end{proposition}
\begin{proof} Suppose $X$ is orthogonal. Then $\begin{ssmatrix} I_n, & X \end{ssmatrix}\begin{ssmatrix} I_n \\ X\T \end{ssmatrix} = I_n + XX\T = I_n+I_n= O_{n\ti n}$. Thus $\begin{ssmatrix} I_n \\ X\T \end{ssmatrix}$, of rank $n$, is a control matrix for the code $\C$ generated by $\begin{ssmatrix} I_n, & X \end{ssmatrix}$. Thus $(\begin{ssmatrix} I_n \\ X\T \end{ssmatrix})\T= \begin{ssmatrix}I_n, & X\end{ssmatrix}$ generates the dual code of $\C$. Hence $\C$ is self-dual.

    On the other hand if the code generated by $(I_n,X)$ is self-dual then $(I,X)\T = \begin{ssmatrix} I_n \\ X\T \end{ssmatrix}$ is a control matrix for this code and so $(I_n,X)\begin{ssmatrix} I_n \\ X\T \end{ssmatrix}= 0_{n\ti n}$. Hence $I_n+XX\T = 0_{n\ti n}$ and so $XX\T = I_n$.  
\end{proof}
\begin{example}\label{self}
The Hamming $[8,4,4]$ self-dual binary code $\mathcal{H}$ is formed this way. Let $U=\begin{ssmatrix}0&1&1&1 \\ 1&1&1&0\\1&1&0&1\\1&0&1&1 \end{ssmatrix}$ Then $U^2=I_4, U=U\T$ and $A=(I_4,U)$ is a generator matrix for the Hamming $[8,4,4]$ self-dual code $\mathcal{H}$. In addition the control matrix for the code has the form $\begin{smatrix}I_{4} \\ U \end{smatrix}$ in which each row is unique and can be used to correct any one error in the one-error correcting code $\mathcal{H}$. 
\end{example}
\begin{example}\label{self1}
The Golay $[24,12,8]$ is formed in this way, \cite{ian}. Let $U$ be the reverse circulant matrix formed using $(0,1,1,0,1,1,1,1,0,1,0,0)$ as the first row. Then $U^2=I_{12},U=U\T$ and $(I_{12},U)$ is a generator matrix for the self-dual Golay $[24,12,8]$ code $\mathcal{G}$. See \cite{ian} for details. The control matrix for the code has the form $\begin{smatrix} I_{12} \\ U\end{smatrix}$.  The sum of any $1$, $2$ or $3$ rows is unique and thus a lookup table can be formed to correct up to three errors in this 3-error correcting Golay $[24,12,8]$ self-dual code.

\end{example}

A generator matrix for a linear block code may be given in systematic form,  $G=(I_n,P)$ where $P$ is an $n\ti t$ matrix, see \cite{blahut}. The distance of the code generated by $(I_n,P)$ is a function of the `unit-derived' type codes from $P$.
\begin{proposition} Consider the code $\C$  generated by $G=(I_n,P)$. Suppose the code generated by any $s$ rows of $P$ has distance $ \geq (d-s)$ and for some choice of $r$ rows the code generated by these $r$ rows has distance exactly $(d-r)$; then the distance of $\C$ is $d$.    
\end{proposition}

The Proposition makes sense even if the number of columns of $P$ is less than $n$.  
The following Lemma is easy from results on fields.

\begin{lemma}\label{sq1} Let $F$ be a field. Then $F$ has an square root of $(-1)$ or else a quadratic extension of $F$ has a square root of $(-1)$.
\end{lemma}
\begin{proof} If $F$ does not have a square root of $(-1)$ then $x^2+1$ is irreducible over $F$.
  \end{proof}
\begin{lemma}\label{orth} Let $F$ be a field with contains a square root of $(-1)$, denoted by $i$, and $X$ an $n\ti n$ matrix over $F$. Then $XX\T = I_n$ if and only if $(iX)(iX)\T = -I_n$.
\end{lemma}
  
Proposition \ref{jut} is implicit in the following more general Proposition which enables the construction of self-dual codes over fields.

\begin{proposition}\label{orthogonalsq} Let $X$ be an $n\ti n$ matrix over a field $\F$.
  
  (i) If $X$ is an orthogonal matrix then $(I_n,iX)$ generates a self-dual code where $i$ is a square root of $(-1)$ in $\F$ or in a quadratic extension of $\F$.

  (ii) If $(I_n,X)$ is self-dual then $iX$ is orthogonal where $i$ is a square root of $(-1)$ in $\F$ or in a quadratic extension of $\F$.
  
\end{proposition}
\begin{proof} (i)  $(I_n,iX)\begin{ssmatrix}I_n \\ (iX)\T \end{ssmatrix} = I_n-XX\T = 0$ and so $(I_n,iX)$ generates a self-dual code as $(\begin{ssmatrix}I_n \\ (iX)\T \end{ssmatrix})\T = (I_n,iX)$.

  (ii) Suppose $(I_n,X)$ is self-dual. Then a control matrix of the code is $(I_n,X)\T = \begin{ssmatrix} I_n \\ X\T \end{ssmatrix}$ and so $(I_n,X)\begin{ssmatrix} I_n \\ X\T \end{ssmatrix}= 0$. Hence $I_n+XX\T = 0$ and so $XX\T = -I_n$. hence $iX(iX)\T = I_n$
  \end{proof}

In a field of characteristic $2$, $(-1)=1$ and so Proposition \ref{jut} follows from Proposition \ref{orthogonalsq}.

This gives a general method for constructing and analysing self-dual codes from unit orthogonal and orthogonal like matrices.  

\subsection{Codes from Fourier type units}\label{fouriertype} 
Using Fourier (or more generally Vandermonde) unit matrices to construct various types of mds codes was initiated in \cite{unitderived,hurleyquantum} and further developed in \cite{hurleyorder} and others. Here we present Propositions in  a very general form from which these constructions may be derived. Series of required types, lengths and rates are achievable. 

Let $F_n$ be a Fourier matrix over a finite field $\F$. Over which finite fields this $F_n$ can be constructed is discussed in \cite{unitderived,hurleyorder} and elsewhere.  Let $F_n^*$ be the inverse of $F_n$ giving the unit scheme $F_nF_n^* = I_n$. Let $F_n$ have rows $\{e_0,e_1, \ldots, e_{n-1}\}$ in order  and $F_n^*$ have columns $\{f_0,f_1, \ldots, f_{n-1}\}$ in order. Now $e_1= (1, \om, \om^2,\ldots, \om^{n-1})$ and $e_i=(1, \om^i, \om^{2i}, \ldots, \om^{(n-1)i})$ where $\om$ is a primitive $n^{th}$ root of unity in the field $\F$.

{\em \bf  Then, see \cite{unitderived,hurleyquantum,hurleyorder}, it is noted that
  $f_i = \f{1}{n}e_{n-i}\T$ and $e_i =nf_{n-1}\T$.}

\begin{proposition}\label{fourier}Let $F_n$ be a Fourier matrix over a finite field $\F$. Let $F_n^*$ be the inverse of $F_n$ giving the unit scheme $F_nF_n^* = I_n$. Let $F_n$ have rows $\{e_0,e_1, \ldots, e_{n-1}\}$ in order  where $e_i=(1, \om^i, \om^{2i}, \ldots, \om^{(n-1)i})$ and  $\om$ is a primitive $n^{th}$ root of unity in the field $\F$.

  Then the basic unit scheme from the Fourier matrix, $F_nF_n^*=I_n$, is given as follows:

  $$\begin{pmatrix} e_0 \\ e_1 \\ \vdots \\ e_{n-1} \end{pmatrix} \begin{pmatrix} e_0\T, & e_{n-1}\T, & e_{n-2}\T, & \ldots, & e_1\T \end{pmatrix} = nI_n$$
\end{proposition}

Having $nI_n$ rather than $I_n$ is no problem in describing codes from the scheme as $n\neq 0$ in a field in which the Fourier $n\ti n$ matrix exists;  $H$ is a control, respectively generator, matrix if and only if $\al H$ is a control, respectively generator, matrix for $\al \neq 0$. 

This gives the following, see for example \cite{unitderived}:
\begin{proposition}\label{fourier1} Let $F_n$ be Fourier $n \ti n$ matrix over a finite field and has rows $\{e_0,e_1, \ldots, e_{n-1}\}$. Suppose $F_n = \begin{ssmatrix} A \\ B \end{ssmatrix}$ where $A= \begin{ssmatrix} e_0 \\ e_1 \\ \vdots \\ e_{r-1} \end{ssmatrix}$.

  (i) The code generated by $A$ is an $[n,r,n-r+1]$ mds block code.

  (ii) When $r>\f{n}{2}$ the code generated by $A$ is an $[n,r,n-r+1]$ DC mds code.

  (iii) In the case when $r> \f{n}{2}$
  the CSS construction from the DC $[n,r,n-r+1]$ code gives a quantum error-correcting $[[n,2r-n, n-r+1]]$ which is an mds quantum code\footnote{mds in the quantum code sense}. 
\end{proposition}

If $r$  rows of $F_n$ are chosen in arithmetic order, starting at any row,  with arithmetic difference $k$ where $\gcd(k,n)=1$ then an mds $[n,r,n-r+1]$ is still obtained; the differences are taken mod $n$. This may be used to construct LCD codes. It will be shown later that DC convolutional codes may in many circumstances be produced from the unit scheme that produced the LCD codes.  Examples are given as follows before the general result is described. 

    \begin{example}\label{ert} Let $F_7$ denote a Fourier $7\ti 7$ matrix over a finite field. Such a matrix exists over a field whose characteristic does not divide $7$ and which has an element of order $7$.
      Thus such a matrix exists for example  over $GF(2^3)$ or over $GF(13^2)$. Look at a unit scheme formed by rearranging the rows of $F_7$ as follows:

      $\begin{ssmatrix}e_6 \\ e_0 \\e_1 \\e_2\\ e_3 \\e_4 \\e_5  \end{ssmatrix}\begin{ssmatrix} e_1\T & e_0\T & e_6\T & e_5\T & e_4\T & e_3\T & e_2\T \end{ssmatrix} = 7I_7$

      Then the first three rows $\begin{ssmatrix}e_6 \\ e_0 \\e_1 \end{ssmatrix}$ generates an $[8,3,6]$ mds code. A control code is $\begin{ssmatrix} e_5\T & e_4\T & e_3\T & e_2\T \end{ssmatrix} $ and thus the dual code is generated by the transpose of this, $\begin{ssmatrix} e_5 \\ e_4 \\  e_3 \\ e2 \end{ssmatrix} $. Thus a $[8,3,6]$ LCD code is obtained; the $e_i$ are independent as rows of an invertible matrix.

      Looking at the unit scheme as given, it will be shown later how a convolutional code which is DC can be constructed.
      \end{example} 

\begin{example}\label{ert1} Let $F_8$ be a Fourier $8\ti 8$ matrix over a finite field. Such a matrix exists in a field of characteristic not dividing $8$  and having an  element of order $8$, for example over $GF(3^2)$ or over the prime field $GF(17)$.

  The rows of $F_8$ in order are denoted by $\{e_0,e_1, \ldots, e_7\}$. Look at the unit-scheme in the form

  $\begin{ssmatrix} e_6 \\ e_7\\ e_0\\ e_1\\ e_2\\e_3\\e_4 \\e_5\end{ssmatrix} \begin{ssmatrix} e_2\T & e_1\T &e_0\T & e_7\T &e_6\T &e_5\T &e_4\T&e_3\T \end{ssmatrix} = 8I_8$.

    Then $\begin{ssmatrix} e_6 \\ e_7\\ e_0\\ e_1\\ e_2\end{ssmatrix}$ generates an $[8,5,4]$ mds code. A control matrix is $\begin{ssmatrix} e_5\T&e_4\T &e_3\T \end{ssmatrix}$ and thus the transpose of this, $\begin{ssmatrix} e_5 \\ e_4 \\ e_3 \end{ssmatrix}$, generates the dual of the code. Hence the code is an LCD code - the $e_i$ are independent as rows of an invertible matrix.
    \end{example}

The idea is to keep a row $e_i$ and its `conjugate' $e_{n-i}$ together in the generating matrix and thus get an LCD code. When using the same unit scheme to obtain convolutional codes, then DC convolutional codes are obtainable. The following is the general result which is best understood by looking at examples \ref{ert} and \ref{ert1}.

\begin{proposition}\label{lcd} Let $F_n$ be a Fourier matrix over a finite field. Let the rows in order of $F_n$ be denoted by $\{e_0,e_1,\ldots, e_{n-1}\}$. Rearrange the Fourier matrix unit scheme  as follows:

  $$\begin{ssmatrix} e_r \\ \vdots \\ e_{n-1} \\e_0 \\e_1 \\ \vdots \\e_{n-r} \\ e_{n-r+1} \\ \vdots \\ e_{r-1}\end{ssmatrix}(e_{n-r}\T,\ldots,e_1\T,e_0\T,e_{n-1}\T,\ldots, e_r\T, e_{r-1}\T,\ldots,e_{n-r+1}\T) = nI_n$$

    Then the code generated by $\begin{ssmatrix} e_r \\ \vdots \\ e_{n-1} \\e_0 \\e_1 \\ \vdots \\e_{n-r} \end{ssmatrix}$ is an LCD mds code.
\end{proposition}

An advantage of this is that looking at the full unit scheme allows the construction of convolutional DC codes and from the convolutional DC code quantum error-correcting codes are defined by the CSS construction.

From full unit schemes:

\noindent DC linear block codes from unit  $\xRightarrow{unit-derived}$ convolutional LCD codes from the unit.

\noindent LCD linear block codes from unit $\xRightarrow{unit-derived}$ convolutional DC codes from the unit $\xRightarrow{unit-derived}$  convolutional quantum codes from the unit scheme.
  
\section{Convolutional  unit-derived codes} The basic unit-derived scheme is given by:
$ \begin{smatrix} A \\ B \end{smatrix} \begin{smatrix}  C & D \end{smatrix}=I_n$.

As noted, using $AD=0$ defines a block linear code where $A$ generates the code and $D\T$ generates  the dual of this code. The total power of the unit is not used as $\{B,C\}$ are ignored. This can be rectified by going on to describe convolutional codes from the unit scheme.
The general idea is to use $A,B$ to describe convolutional codes $G(z) = A + Rz$ where $R$ is formed from  $B$. 
DC  and LCD  convolutional codes are  obtained. Constructing DC convolutional  codes lead to the construction of quantum error correcting convolutional, QECC, codes, by CSS construction. 

A free distance can  be prescribed as a linear functional of the distances of the (block linear) codes generated by $A$ and $B$.

Convolutional mds codes are constructed in \cite{hurleyunit,hurleyorder} by the method. 

\subsection{Same block sizes}
Suppose that $A,B$ both have the same size, $r\ti n$, in the unit-derived formula; now $r=\f{n}{2}$ and $n$ is even. Let  the code generated by $A$ have distance $d_1$ and the code generated by $B$ have distance $d_2$. Consider $G[z] = A+Bz$. This generates a convolutional code of memory $1$. As $G(z)*C= AC= I_r$ the generator matrix has a right inverse and so the code is non-catastrophic.

Now also $(A+Bz)(D-Cz) = AD -ACz +BDz- BCz = -I_r z+I_rz = 0_r$. Thus $D-Cz$ is a control matrix for the $(n,r,r;1) $ convolutional code. 

The free distance for the code is $(d_1+d_2)$. This minimum distance is obtained when the information vector has support $1$.  
If the information vector $P(z)$ has support $k$ then $P(z)G(z)$ has distance $ \geq (d_1+d_2+k-1)$.  
  
  The dual code generator matrix is obtained from the control matrix $H\T(z)$; as noted the dual generator matrix is $H(z^{-1})z^m$ where $m$ is the memory and $H\T(z)$ is the control matrix. In this case the control matrix is $H\T(z)= D-Cz$  and so a generator matrix for the dual code is $(D\T-C\T z^{-1})z = -C\T + D\T z$.

  If $C\T = -A$ and $ D\T=B$  then a self-dual convolutional code  is obtained. Such a situation arises when  $ \begin{smatrix} A \\ B \end{smatrix}$ is orthogonal 
  and  the characteristic is $2$. 

  As noted in section \ref{secorthog}, Lemma \ref{sq1}, a finite field $F$  has a square root of $(-1)$ or else a quadratic extension of $F$ has a square root. Thus define $G(z) = A + i B$ where $i$ denotes a square root of $(-1)$. Then $(A+iB)(iD+Cz) = 0$ and so $H\T (z) = iD + Cz$ is a control matrix giving that $H(z^{-1})z = C\T + iD\T$ is a dual matrix. In case $U$ is an orthogonal matrix, $\C\T=A, D\T=B$ and a dual  matrix is $A+iB$ giving that $G(z)$ is a self-dual convolutional code of distance equal to the sum of the distances of the codes generated by $A$ and by $B$.    
  
  \begin{proposition}\label{conv1} Let $U$ be a $2n\ti 2n$ invertible matrix. Suppose $U= \begin{smatrix}A \\ B\end{smatrix}$ and $\begin{smatrix}A \\ B\end{smatrix}\begin{smatrix}  C & D \end{smatrix}=I_{2n}$ where $A$ and $B$ have size $n\ti 2n$ and  $C,D$ have size $2n\ti n$. 

        \begin{enumerate}

\item   $A$ generates a $[2n,n]$ code $\C$ and $D\T$ generates the dual code of $\C$. $B$ generates a $[2n,n]$ code $\D$ and $C\T$ generates the dual code of $\D$.

 \item   $G(z) = A+Bz$ generates a (non-catastrophic)  convolutional $(2n,n,n;1)$ code $\C$. Then $G(z)(D-Cz) =0$, $D-Cz$ is a control matrix of $\C$ and $-C\T+D\T z$ generates the dual code of $\C$.

\item If  the code generated by $A$ has distance $d_1$ and the code generated by $B$ has distance $d_2$, then the (free) distance of $\C$ is $d=d_1+d_2$ and $\C$ is a $(2n,n,n;1,d)$ convolutional code. Further if the information vector $P(z)$ has support $k$ then $P(z)G(z)$ has distance $\geq (d+k-1)$. 
\end{enumerate}
        \end{proposition} 

  \begin{proposition}\label{conv2} Let $U$ be a $2n\ti 2n$ orthogonal matrix. Suppose $U= \begin{smatrix}A \\ B\end{smatrix}$ and  $\begin{smatrix}A \\ B\end{smatrix}\begin{smatrix}  C & D \end{smatrix}=I_{2n}$ where $A$ and $B$ have size $n\ti 2n$  and $C,D$ have size $2n\ti n$.
        
  \begin{enumerate}    \item  $A$ generates a $[2n,n]$ code $\C$ and $B$ generates the dual of $\C$; hence $\C$  is an LCD code.  $B$ generates a $[2n,n]$ code $\D$ and $A$ generates the dual of $\D$; hence $\D$ is an LCD code. 
\item 
   $G(z) = A+Bz$ generates a (non-catastrophic)  convolutional $(2n,n,n;1)$ code $\C$. Then  $-C\T+D\T z=-A+Bz$ generates the dual code of $\C$.

  \item If $U$ is a matrix over a field of characteristic $2$ then $G(z) = A+Bz$ generates a self-dual convolutional code. 

    \item\label{df} If  the code generated by $A$ has distance $d_1$ and the code generated by $B$ has distance $d_2$, then the (free) distance of the code generated by $G(z)$ is $d=d_1+d_2$ and is a $(2n,n,n;1,d)$ convolutional code. Further if the information vector $P(z)$ has support $k$ then $P(z)G(z)$ has distance $\geq (d+k-1)$.

    \item In case of characteristic $2$, the code generated by $G(z)$ is used to generate a quantum convolutional code of memory $1$ which has type $[[2n,0,d]]$ where $d=d_1+d_2$ is given in item \ref{df}.
    \item\label{ds} Define $G(z) = A + i B$ where $i$ denotes a square root of $(-1)$. Then $(A+iB)(iD+Cz) = 0$ and so $H\T (z) = iD + Cz$ is a control matrix giving that $H(z^{-1})z = C\T + iD\T z$ is a dual matrix. The free distance of $G(z)$ is equal to the sum, $d$, of the distances of the codes generated by $A$ and by $B$.    
 
    \item Suppose now $U=\begin{ssmatrix} A \\ B \end{ssmatrix}$ is an orthogonal matrix, then in item \ref{ds}, $C\T=A, D\T=B$. A dual  matrix is then $A+iBz$ giving that $G(z)$ is a self-dual convolutional $(2n,n,n;1,d)$ code. This can be used to define a $[[2n,0,d]]$ quantum convolutional code.  \end{enumerate}
  
  \end{proposition} 

  These propositions are very general and can be used to construct infinite series of such codes with increasing distances. 
\subsection{Different block sizes}
  Consider cases where $A$ has size greater than $B$ in the unit-derived formula $ \begin{ssmatrix} A \\ B \end{ssmatrix} \begin{ssmatrix} C & D \end{ssmatrix} = I_n$ with $U=\begin{ssmatrix} A \\ B \end{ssmatrix}$.

  Let   $A$ have size $r\ti n$, then $B$ has size $(n-r)\ti n$.  Now  $r > (n-r)$  is equivalent to $2r>n$. Let $t=r-(n-r) =2r-n$ and  $0_t=0_{t\ti n}$. Thus
  $B_1 = \begin{pmatrix} 0_t \\ B\end{pmatrix}$ is an $r\ti n$ matrix. Now $C$ is an $n\ti r$ matrix and thus has the form $C=(X,C_1)$ where $C_1$ has size $n\ti (n-r)$ and $X$ has size $n\ti (2r-n)$. As $AC = I_r$ then $AC_1= \begin{pmatrix} 0_{(2r-n) \ti (n-r)} \\ I_{(n-r)\ti (n-r)}\end{pmatrix}$. 

    Define $G(z) = A + B_1z$. This defines a generator matrix for a convolutional $(n,r,n-r;1)$ code. $BC=0$ implies $BC_1=0_{(n-r)\ti (n-r)}$. 

    Then $(A+B_1z)(D-C_1z)=AD-AC_1z+B_1Dz -B_1C_1z^2 = 0_{r\ti (n-r)} -\begin{pmatrix} 0_{2r-n \ti n-r} \\ I_{n-r\ti n-r}\end{pmatrix}z + \begin{pmatrix} 0_{2r-n \ti n-r} \\ I_{(n-r)\ti (n-r)}\end{pmatrix}z = 0_{r\ti (n-r)}$. Thus the control matrix is $(D-C_1z)$ and a dual matrix is $-C\T+Dz$ As $(A+B_1z)C=I_r$, the code is non-catastrophic.

    Define $G(z) = A+iB_1z$ where $i$ is a square root of $(-1)$ in the field or in a quadratic extension of the field. Then $(A+iB_1)(iD+C_1z) = 0$ and $C_1\T + iD\T z$ is a dual matrix. In case $U$ is orthogonal $C= A\T, D=B\T$ and $A= \begin{ssmatrix}X\T \\ C_1\T\end{ssmatrix}, B_1=\begin{ssmatrix} 0 \\ D\T \end{ssmatrix}$. Hence the code generated by $G(z)$ contains its dual.
    \begin{proposition}\label{conv3} Let $U$ be a matrix unit over a field with $U= \begin{smatrix}A \\ B\end{smatrix}$ where $A$ has size $r\ti n$ and $B$ has size $(n-r)\ti n$ with $r>n-r$. Let $t=2r-n$ and $B_1= \begin{smatrix} 0_{t\ti n} \\ B\end{smatrix}$. Then

          (1) $G(z) = A+B_1z$ generates a convolutional $(n,r,n-r;1)$ code $\C$.

          (2) Let $A_1$ be the matrix of the first $(2r-n)$ rows of $A$. The distance $d$ of $\C$ is $\min\{d(A_1), d(A) + d(\begin{ssmatrix} A_1 \\ B \end{ssmatrix})\}$ where $d(X)$ denotes the distance of the code generated by $X$. 

          
\end{proposition} 

    \begin{proposition}\label{conv4} Let $U$ be an orthogonal matrix in a field $\F$ and  $U= \begin{smatrix}A \\ B_1\end{smatrix}$ where $A$ has size $r\ti n$ and $B_1$ has size $(n-r)\ti n$ with $r>n-r$. 
Let $t=2r-n$ and $B_1= \begin{smatrix} 0_{t\ti n} \\ B\end{smatrix}$. Let $i$ be a square root of $-1$ in $F$ or in a quadratic extension of $\F$ Then

          (1) $G(z) = A+i B_1z$ generates a convolutional dual-containing $(n,r,n-r;1)$ code $\C$.

          (2) Let $A_1$ be the matrix of the first $(2r-n)$ rows of $A$. The distance $d$ of $\C$ is $\min\{d(A_1), d(A) + d(\begin{ssmatrix} A_1 \\ B \end{ssmatrix})\}$ where $d(X)$ denotes the distance of the code generated by $X$. 
(3) A quantum convolutional code of the form $[[n,2r-n,d]]$ is constructed from $\C$ where $d$ is the distance of $\C$. 

      
    \end{proposition}

    The process is developed similarly by looking at
    $ \begin{ssmatrix} A \\ B \end{ssmatrix} \begin{ssmatrix} C & D \end{ssmatrix} = \al I_n$ with $U=\begin{ssmatrix} A \\ B \end{ssmatrix}$ where $\al \neq 0$. See for example section \ref{hadamard} for examples on this. 
    
    It is best illustrated by looking at a block decomposition say one of the form $\begin{ssmatrix} A_0 \\ A_1\\A_2\\ A_3\end{ssmatrix}\begin{ssmatrix}B_0 &B_1&B_2&B_3\end{ssmatrix} = I_{4n}$ where each $A_i$ has size $n\ti 4n$. Take $G(z) =\begin{ssmatrix} A_0 \\ A_1\\A_2\end{ssmatrix}+\begin{ssmatrix}0 \\ 0\\ A_3\end{ssmatrix}z$.

        In case  $U$ is orthogonal, let  $G(z) =\begin{ssmatrix} A_0 \\ A_1\\A_2\end{ssmatrix}+i\begin{ssmatrix}0 \\ 0\\ A_3\end{ssmatrix}z$ and then a $(4n,3n,n;1,d)$ DC convolutional code is obtained from which a quantum convolutional code of form $[[4n,2n,d]]$ is obtained. The $d$ may be calculated algebraically and depends on the distances of codes formed from the blocks.
    
    \subsubsection{Prototype examples}

    The following prototype examples exemplify some of the general constructions.
    The examples given tend to be  linear block and  convolutional of types  DC and LCD, but many types and infinite series of such may be built up also by the techniques. The types DC lead to the formation of quantum codes 

    

    \begin{example}\label{first} {\em Let $F_7$ be a Fourier $7\ti 7$ matrix over some field $\F$. The field $\F$ is any field over which the Fourier $7\ti 7$ matrix exists. Thus $\F$ could be $GF(2^3)$, a characteristic $2$ field, but also over fields with characteristic not dividing $7$.

      Denote the rows of $F_7$ in order by $\{e_0,e_1,e_2,e_3,e_4,e_5,e_6\}$ and the columns of the inverse of $F_7$ in order by $\{f_0,f_1,f_2,f_3,f_4,f_5,f_6\}$. Note that $e_if_j = 0, i\neq j, e_if_i = 1$ but also as $F_7$ is a Fourier matrix  that $f_i\T = \f{1}{7}e_{7-i}$. The fraction part is no problem for check or control matrices: $H$ is a check matrix if and only if $\al H$ is a check matrix for any $\al \neq 0$. 
      
      Let $A=\begin{ssmatrix}e_0 \\e_1 \\ e_2 \\e_3 \end{ssmatrix}, B_1=\begin{ssmatrix}e_4\\e_5\\e_6\end{ssmatrix}$ where $F_7=\begin{pmatrix} A \\ B_1 \end{pmatrix}$ is the Fourier $7\ti 7$ matrix. 
      Now from \cite{hurleyunit} $A$ generates a $[7,4,4]$ mds (block) DC code. 

      Let $B=\begin{smatrix}0\\B_1 \end{smatrix}$.
            Define $G(z)= A + Bz = A+\begin{ssmatrix} 0 \\ e_4\\ e_5\\e_6 \end{ssmatrix}z$. Now $G(z)*(f_1,f_2,f_3,f_4)=I_4$ and so $G(z)$ has a right inverse and thus the code generated by $G(z)$ is non-catastrophic.  $G(z)$ defines a $[7,4,3;1]$ convolutional code $\C$. The GSB for such a code is $(7-4)(\floor{\f{3}{4}}+1)+ 3+1= 3+3+1=7$. It is  easy to check that this $7$  is the free distance of the code and so the code is an mds convolutional code.

      Note that $(\al_1,\al_2,\al_3,\al_4)*B$ has distance $\geq 5$ as an element in a $[7,3,5]$ code where $(\al_1,\al_2,\al_3,\al_4)$ is a non-zero $1\ti 4$ vector,  except when $(\al_1,\al_2,\al_3,\al_4)= (\al_1,0,0,0)$; but then $(\al_1,0,0,0)*A = \al_1e_0$ which has distance $7$.  


      $G(z)*((f_4,f_5,f_6) - (f_2,f_3,f_4)z) = 0$ and so $H\T(z)=(f_4,f_5,f_6) - (f_2,f_3,f_4)z$ is a control matrix. Then $H(z^{-1})z =\begin{ssmatrix}f_4\T \\ f_5\T\\f_6\T\end{ssmatrix}z - \begin{ssmatrix}f_1\T\\f_2\T\\f_3\T \end{ssmatrix}$ generates the dual matrix. Now since $f_i\T=\f{1}{7}e_{7-i}$ this means $7*H(z^{-1})z= -\begin{ssmatrix}e_6\\e_5\\e_4\end{ssmatrix} + \begin{smatrix}e_3\\e_2\\e_1\end{smatrix}z$ generates the dual matrix. Thus  $\C$ is a convolutional $(7,4,3;1,7)$ code which is an LCD  code.} 
      
    \end{example}
    \begin{example} Use the same setup as in Example \ref{first} where  $F_7$ is a Fourier $7\ti 7$ matrix. Take $A=\begin{ssmatrix}e_0\\e_1\\e_6\\e_2\\e_5\end{ssmatrix}$. Then it follows from \cite{hurleyunit} that the code generated by $A$ is an LCD $[7,5,3]$ code.  Let $B=\begin{ssmatrix} 0 \\0\\0\\e_4 \\e_3\end{ssmatrix}$

     Define $ G(z)=A +Bz = \begin{ssmatrix}e_0\\e_1\\e_6\\e_2\\e_5\end{ssmatrix} + \begin{ssmatrix} 0 \\0\\0\\e_4 \\e_3\end{ssmatrix}z$.
          Then $(A+Bz)*(f_0,f_1,f_6,f_2,f_5)= I_5$ and so $G(z)$ has a right inverse and hence code, $\C$, generated by $G(z)$ is non-catastrophic.
            Now $G(z)*\{(f_4,f_3)-(f_2,f_5)z\} = 0$ and so $(f_4,f_3)-(f_2,f_5)z= H\T(z)$ is a control matrix.

                Now $H(z^{-1})z = \begin{ssmatrix}f_4\T \\ f_3\T\end{ssmatrix}z
                  -\begin{ssmatrix}f_2\T \\ f_5\T \end{ssmatrix}=
                  \f{1}{7}\{-\begin{ssmatrix}e_5 \\ e_2 \end{ssmatrix}+ \begin{ssmatrix}e_3\\e_4 \end{ssmatrix}z\}$.

                  Thus $-\begin{ssmatrix}e_5 \\ e_2 \end{ssmatrix}+ \begin{ssmatrix}e_3\\e_4 \end{ssmatrix}z$ generates the dual code of $\C$.

                  If the field $\F$ has characteristic $2$ the code $\C$ is dual containing. Thus when $F_7$ is a Fourier matrix over $GF(2^3)$, the code $\C$ is dual containing. It is easy to check directly that the free distance of $\C$ is $5$ and is thus an mds convolutional $(7,5,2;1,5)$ code; this  is dual-containing when $F_7$ is a Fourier matrix  over $GF(2^3)$.

                  If the field does not have characteristic $2$ then define $G(z) = A + iBz$ where $i$ is a square root of $(-1)$ in the field or in a quadratic extension of the field. Then again a mds convolutional dual-containing code is obtained. A quantum convolutional code is constructed from a dual-containing code. 

    \end{example}



  The next example is a prototype example which although small demonstrates the power of the methods. In larger examples each row is replaced by a block of rows and lengths, distances are  increased substantially. 
\begin{example}
  Let $X= \begin{ssmatrix}0&1&1&1\\1&1&1&0\\1&1&0&1\\1&0&1&1\end{ssmatrix}$ over $GF(2)$. Then $X^2=I_4, X=X\T$. Thus DC codes are obtained by taking rows of $X$ as a generating matrix  and deleting corresponding columns of $X\T =X$ to obtain  a control matrix.  

Let $A= \begin{ssmatrix}0&1&1&1\\1&1&1&0 \end{ssmatrix}, B=  \begin{ssmatrix}1&1&0&1\\1&0&1&1\end{ssmatrix}$ and 
  $\begin{smatrix} A \\ B \end{smatrix}  \begin{smatrix}  C & D \end{smatrix}=I_4$ is our unit scheme. Now here $D=B\T, C=A\T$. Thus $A$ generates a $[4,2,2]$ code $\C$ with control matrix $D=B\T$ and so $B$ generates the dual code of $\C$. The code $\C$ is an $[4,2,2]$ LCD code.

  Extend this to a convolutional code $\C$ using $G(z)=A+Bz$. Now $(A+Bz)(D+Cz)=0$ so that $H\T(z)=D+Cz$ is a control matrix. Also $G(z)* C= I_2$ and so the code is non-catastrophic.
  The dual code of $\C$ is generated by $H(z^{-1})z = C\T+D\T z = A+Bz$ and hence the code is self-dual. Thus a convolutional self-dual $(4,2,2;1,4)$ is obtained. From this a quantum error-correcting code of form $[[4,0,4]]$ is obtained.

   $P(z)G(z)$ has distance $\geq 4+(s-1)$ for an  information vector $P(z)$   of support $s$.

  Any rows of $U$ may be chosen to generate a code and the resulting code is automatically an LCD code. 
  For example choose the first and third row of $U$ and get $A= \begin{ssmatrix} 0&1&1&1\\ 1&1&0&1\end{ssmatrix}, B= \begin{ssmatrix}1&1&1&0\\1&0&1&1 \end{ssmatrix}$. Then $D=B\T, C=A\T$ similar to above. $A$ then generates an LCD $[4,2,2]$ and $G(z) = A+Bz$ generates a non-catastrophic self-dual $(4,2,2;1,4)$ convolutional code.
 
    This idea of choosing arbitrary rows,  when used on large size matrices, lends itself to forming McEliece type of cryptographic systems.

    Larger rates are obtained as follows. Choose three rows of $U$ and get
    $A=\begin{ssmatrix}0&1&1&1\\1&1&1&0\\1&1&0&1\end{ssmatrix}, B=\begin{ssmatrix}1&0&1&1\end{ssmatrix}$. Then in general form $C=A\T, D=B\T$. $A$ generates a $[4,3,1]$ LCD code. Give  $U$ in its row form: $U= \begin{ssmatrix}E_0\\E_1\\E_2\\E_3\end{ssmatrix}$. (In larger length constructions the $E_i$ are blocks of matrices of size $n\ti 4n$.) Then $A=\begin{ssmatrix}E_0\\E_1\\E_2\end{ssmatrix}, B=\begin{ssmatrix}E_3\end{ssmatrix}$.

      Define $G(z) = \begin{ssmatrix}E_0\\E_1\\E_2\end{ssmatrix} + \begin{ssmatrix}0_4\\0_4\\E_3\end{ssmatrix}z = A+B_1z$, say; $0_4$ indicates here a row of zeros of length $4$. This gives a $(4,3,1;1)$ convolutional code which is dual-containing. The distance is $2$. Note that $(\al_1,\al_2,\al_3)*B_1= \al_3e_3$ has distance $3$ except when $\al_3 = 0$ in which case $(\al_1,\al_2,\al_3)A = \al_1e_0+\al_2e_1$ has distance $2$ or $3$. If $P(z)$ is an information vector of support $s$ then $P(z)G(z)$ has distance $\geq 2+(s-1)$. 
   \end{example}

          \begin{example}
          In another way construct rate $\f{3}{4}$ and $\f{1}{4}$ convolutional codes as follows:

Define          $G(z) = \begin{ssmatrix}E_0\\E_1\\E_2\end{ssmatrix} +
              \begin{ssmatrix}E_1\\E_0\\E_3\end{ssmatrix}z+\begin{ssmatrix}E_2\\E_3\\E_0\end{ssmatrix}z^2+\begin{ssmatrix}E_3\\E_2\\E_2\end{ssmatrix}z^3$
and  $H\T(z)=E_3\T+E_2\T z+E_1\T z^2 + E_0\T z^3$.

                Then $G(z)H\T(z) = 0$. Let the code generated by $G(z)$ be denoted by $\C$. The dual of $\C$ is generated by $H(z^{-1})z^3 = E_0+E_1z+E_2z^2+E_3z^3$. Thus $\C$ is a dual containing $(4,3,9;3)$ convolutional code. Its free distance is $4$ giving a $(4,3,9;3,4)$ convolutional dual-containing code. From this a quantum $[[4,2,4]]$ convolutional code is obtained.

$P(z)*G(z)$ has distance $\geq 4+(s-1)$ when $P(z)$ is an information vector of support $s$. 
\end{example}

                \begin{example}\label{hamm} Hamming convolutional code: Here the Hamming $[7,4,3]$ code is extended to a  $(7,4,3;1,6)$ convolutional code. The distance is $6$ which is twice that of the Hamming $[7,4,3]$ but it's also a convolutional code which has its own decoding techniques. The method is to look at the Hamming code as a unit-derived code and proceed from there by a general technique of constructing convolutional codes by the unit-derived method.

                The Hamming $[7,4,3]$ is  given with generator matrix $G=\begin{ssmatrix}1 &0&0&0&1&1&0 \\0&1&0&0&1&0&1\\0&0&1&0&0&1&1\\0&0&0&1&1&1&1\end{ssmatrix}$ and check matrix $H=\begin{ssmatrix}1&1&0&1&1&0&0\\1&0&1&1&0&1&0\\0&1&1&1&0&0&1 \end{ssmatrix}$. Now $\begin{ssmatrix} G \\ H\end{ssmatrix}$ is {\em not} invertible so this cannot be used for extending $G$ to be a unit-derived code. For reasons that will appear later use $L= \begin{ssmatrix}1 &1&1&1&1&1&1 \\0&1&0&0&1&0&1\\0&0&1&0&0&1&1\\0&0&0&1&1&1&1\end{ssmatrix}$ as the generator matrix\footnote{This is obtained by adding the other three rows to the first row to $G$ above}.

      Now complete $L$ to a unit $U=\begin{ssmatrix}1 &1&1&1&1&1&1 \\0&1&0&0&1&0&1\\0&0&1&0&0&1&1\\0&0&0&1&1&1&1 \\1&0&1&1&1&0&0\\0&1&0&0&1&1&1\\0&0&0&1&1&1&0\end{ssmatrix} = \begin{ssmatrix}L \\ K\end{ssmatrix}$, say. It is easy to check that $K$ generates a $[7,3,3]$ code.

        $U$  has inverse $V=\begin{ssmatrix}0&0&1&1&1&0&0\\1&1&0&1&1&1&1\\0&1&1&1&0&1&1\\1&1&0&0&1&0&1\\1&0&0&0&1&1&0\\0&1&0&0&0&1&0\\0&0&0&1&0&0&1\end{ssmatrix} = \begin{ssmatrix} C & D \end{ssmatrix}$,  where $C$ is $7\ti 4$ and $D$ is $7\ti 3$.

        Form $G(z)=L+\begin{ssmatrix}0\\K\end{ssmatrix}z = \begin{ssmatrix}1 &1&1&1&1&1&1 \\0&1&0&0&1&0&1\\0&0&1&0&0&1&1\\0&0&0&1&1&1&1\end{ssmatrix}+
          \begin{ssmatrix}0&0&0&0&0&0&0 \\1&0&1&1&1&0&0\\0&1&0&0&1&1&1\\0&0&0&1&1&1&0\end{ssmatrix}z$.

 Precisely:  \begin{itemize}\label{join} \item $G(z)$ generates  a convolutional $(7,4,3;1)$ non-catastrophic code $\C$.

              \item The free distance of $\C$ is $6$ so $G(z)$ generates  a $(7,4,3;1,6)$ convolutional code.

              \item  If $P(z)$ is an information vector 
              then the distance of $P(z)*G(z)$ is $\geq (6 +d-1)$ where $d$ is the support of $P(z)$.

            \item  The free distance may be shown from the following observations.
 Let $(\al_1,\al_2,\al_3,\al_4)$ be a non-zero vector. Then $(\al_1,\al_2,\al_3,\al_4)*\begin{ssmatrix} 0 \\ K
 \end{ssmatrix}$ has distance $3$ except when $\al_2=0=\al_3=\al_4$; but in this case $(\al_1,\al_2,\al_3,\al_4)*L = \al_1(1,1,1,1,1,1,1)$ which has distance $7$.
 \end{itemize}

Term the code generated by $G(z)$ to be  the {\em Hamming convolutional code}. 
\end{example}

\section{Higher memory convolutional codes from units}
The basic unit-derived scheme from $UV=I$ breaks $U=\begin{ssmatrix}A \\ B \end{ssmatrix}$ to derive $\begin{ssmatrix} A \\ B\end{ssmatrix}\begin{ssmatrix} C & D \end{ssmatrix}$ and  linear block and convolutional codes of memory $1$ from the scheme have been described and analysed.

  Here the unit is broken into more than two  blocks and linear block and convolutional codes of high memory are derived and analysed.

  Consider the case $UV=I$ where $U=\begin{ssmatrix} A \\ B \\ C \\D \end{ssmatrix}, V= \begin{ssmatrix} E & F & G & H \end{ssmatrix}$ appropriately, giving another type of (basic) unit scheme:

  $$\begin{ssmatrix}A\\B\\C\\D\end{ssmatrix}\begin{ssmatrix} E & F & G & H \end{ssmatrix} =I$$
  
   First  assume the sizes of $A,B,C,D$ are the same. Thus $U$ is a $4n\ti 4n$ matrix. Three-quarter rate block linear codes are described by taking $\begin{ssmatrix}A\\B\\C\end{ssmatrix}$ as a generator matrix and then $\begin{ssmatrix} H \end{ssmatrix}$ is a control matrix. More generally by choosing three of $A,B,C,D$ to form the generator matrix for a code gives a $[4n,3n]$ three quarter rate code. The control matrix is immediately clear and is one of $E,F,G,H$. 
     When $U$ is orthogonal, LCD block codes are obtained from which the convolutional codes described are DC when the characteristic is $2$.
      
      Memory $3$ codes are described: $G(z)=A+Bz+Cz^2+Dz^3$ is the generator of a $(4n,n,3n;3)$ code. This is non-catastrophic as $G(z)E=I_n$. The distance is $d$ where $d$ is a linear functional of the distances of the codes generated by $A,B,C,D$. Moreover $P(z)G(z)$ has distance $\geq (d+ t-1)$ where $P(z)$ is an information vector and $t$ is the support of $P(z)$.

      Then $(A+Bz+Cz^2+Dz^3)((F,G,H)-(E,H,G)z-(H,E,F)z^2+(G,F,E)z^3 = 0$ and so $K\T(z)=(F,G,H)-(E,H,G)z-(H,E,F)z^2+(G,F,E)z^3$ is a control matrix. The matrix of the dual is given by
      
      $K(z^{-1})z^3 = \begin{ssmatrix}G\T \\ F\T \\ E\T \end{ssmatrix}- \begin{ssmatrix}H\T \\E\T \\F\T \end{ssmatrix}z - \begin{ssmatrix}E\T \\ H\T \\ G\T \end{ssmatrix}z^2 + \begin{ssmatrix}F\T \\ G\T \\ H\T \end{ssmatrix}z^3$.

      This dual code is a $(4n,3n,9n;3)$ code. When $U$ is orthogonal, $A=E\T, B=F\T, C=G\T, D=H\T$.

      \begin{proposition}\label{mem} Let $\begin{ssmatrix}A\\B\\C\\D\end{ssmatrix}\begin{ssmatrix} E & F & G & H \end{ssmatrix} =I_{4n}$ be a unit scheme in which  $\{A,B,C,D\}$ are of the same size. Then

          (i) $G(z)=A+Bz+Cz^2+Dz^3$ is a  generator matrix of a $(4n,n,3n;3)$ convolutional code. The distance is a linear  functional  of the distances of the codes generated by $\{A,B,C,D\}$.
          
          (ii) $P(z)G(z)$ has distance $\geq (d+ t-1)$ where $t$ is the support of the information vector $P(z)$.

          (iii) The control matrix  of $\C$ is  $(F,G,H)-(E,H,G)z-(H,E,F)z^2+(G,F,E)z^3$ and the dual code of $\C$ is generated by $\begin{ssmatrix}G\T \\ F\T \\ E\T \end{ssmatrix}- \begin{ssmatrix}H\T \\E\T \\F\T \end{ssmatrix}z - \begin{ssmatrix}E\T \\ H\T \\ G\T \end{ssmatrix}z^2 + \begin{ssmatrix}F\T \\ G\T \\ H\T \end{ssmatrix}z^3$.

          (iv) When the full matrix is orthogonal the dual code of $\C$  is generated by 

          $\begin{ssmatrix}C \\ B \\ A \end{ssmatrix}- \begin{ssmatrix}D \\ A \\B \end{ssmatrix}z - \begin{ssmatrix}A \\ D \\ C \end{ssmatrix}z^2 + \begin{ssmatrix}B \\ C \\ D \end{ssmatrix}z^3$.

          (v) When the full matrix is orthogonal and the characteristic is $2$ the dual code  is generated by
  $\begin{ssmatrix}C \\ B \\ A \end{ssmatrix}+ \begin{ssmatrix}D \\ A \\B \end{ssmatrix}z + \begin{ssmatrix}A \\ D \\ C \end{ssmatrix}z^2 + \begin{ssmatrix}B \\ C \\ D \end{ssmatrix}z^3$. In this case the dual code $\C^\perp$ of $\C$ is a dual containing $(4n,3n, 9n;3)$ convolutional code. From this dual containing code of rate $\f{3}{4}$, using the CSS construction, a quantum error correcting code of length $4n$ and rate $\f{1}{2}$ is obtained.
          \end{proposition}

         The example \ref{first1} below is a very small prototype example with which to illustrate the general method. 
      \begin{example}\label{first1} 
        Consider $X= \begin{ssmatrix} 0& 1&1&1 \\ 1&1&1&0 \\ 1&1&0&1 \\ 1&0&1&1 \end{ssmatrix}$ over $GF(2)$. The matrix is orthogonal, $XX\T = I_4$, and also $X=X\T$.\footnote{$(I_4,X)$ is a generator matrix for the extended Hamming $[8,4,4]$ code.} Then $XX\T=I$ is broken up to give 
        $\begin{ssmatrix} A \\ B \\ C \\D \end{ssmatrix} \begin{ssmatrix} E & F & G & H \end{ssmatrix}=I_4$ where $\{A,B,C,D\}$ are row $1\ti 4$ vectors and $E\T=A,F\T=B, G\T=C, H\T=D$ as $X$ is orthogonal. Define $G(z) =A +Bz+Cz^2+Dz^3$. Then $G(z)$ generates  a $(4,1, 1;3,12)$ convolutional (non-catastrophic) code $\C$. The distance is $12$ as  the code generated by each of $A,B,C,D$ has distance $3$.

        By Proposition \ref{mem} part (iv), the dual of $\C$ is generated by
        $\begin{ssmatrix}C \\ B \\ A \end{ssmatrix}+ \begin{ssmatrix}D \\ A \\B \end{ssmatrix}z + \begin{ssmatrix}A \\ D \\ C \end{ssmatrix}z^2 + \begin{ssmatrix}B \\ C \\ D \end{ssmatrix}z^3$.

        Thus $\C^\perp$ is a dual-containing convolutional rate $\f{3}{4}$ code of the form $(4,3,9;3)$. From this a quantum error-correcting code of rate $\f{1}{2}$ is formed.

        \end{example}
      \begin{example}\label{unit1} {Golay binary code to convolutional rates $\f{3}{4}$ and $\f{1}{4}$ codes with memory $3$.}

        Consider the matrix $X$ used in forming the self-dual Golay binary $[24,12,8]$ code in the form $(I_{12},X)$ as in \cite{ian}. This $X$ is the reverse circulant matrix with first row $L=[0,1,1,0,1,1,1,1,0,1,0,0]$. The $X$ is symmetric and $XX\T = X^2 = I_{12}$. Here break $X$ into four blocks, $X_1,X_2,X_3,X_4$ of equal size $3\ti 12$. The code generated by each $X_i$ has distance $5$.  Then define $G(z)=X_1+X_2z+X_3z^2+X_4z^3$ and $G(z)$ generates a binary $(12,3, 9;3, 20)$ code $\C$. Also $P(z)G(z)$ has distance $\geq (20 + s-1)$ where $P(z)$ is an information vector and $s$ is the support of $P(z)$. Note that since the rows  of $X$ are independent so any non-zero combination of the rows of $X_1\cup X_2\cup X_3\cup X_4$ has distance $\geq 1$.

          As $X$ is orthogonal the dual, $\C^\perp$,  of $\C$ is  generated 
          by $\begin{ssmatrix}X_3\\ X_2\\X_1 \end{ssmatrix}+\begin{ssmatrix}X_4\\ X_1\\X_2 \end{ssmatrix}z+\begin{ssmatrix}X_1\\ X_4\\X_3 \end{ssmatrix}z^2+\begin{ssmatrix}X_2\\ X_3\\X_4 \end{ssmatrix}z^3$ by Proposition \ref{mem}. Thus $\C^\perp$ is a convolutional $(12,9,9;3)$ code. It is  seen that  $\C^\perp$ is a dual containing convolutional code of rate $ \f{3}{4}$ and is used to form a convolutional quantum error correcting code of rate $\f{1}{2}$. 
          \end{example}

      \subsection{Further block decompositions}
      \begin{enumerate} \item Cases where the unit system is of the form 
      A unit system $\begin{ssmatrix}A\\B\\C\\D\end{ssmatrix}\begin{ssmatrix} E & F & G & H \end{ssmatrix} =I$ where $A,B,C$ have the same size $r\ti n$ but $D$ has size $s\ti n$ with $s<n$ can be dealt in a similar but more complicated manner. 
        The details are omitted.

\item        Cases where the unit system is of size $3n\ti 3n$ is dealt with by the following Proposition:
\begin{proposition}\label{three}  Let   $\begin{ssmatrix}A\\B\\C \end{ssmatrix}\begin{ssmatrix}D& E & F \end{ssmatrix} =I_{3n}$ be a unit scheme in which $A,B,C$ are of the same size. 
   Let $G(z) = A+Bz+Cz^2$. Then  and then verifying that 
          $(A+Bz+Cz^2)( (E,F) - (D,F)z+ (D,E)z^2 + (0, E-D)z^3) = 0$ \end{proposition} This allows the construction of rate $\f{1}{3}$ and rate $\f{2}{3}$ convolutional codes which are dual to one another similar to the method and results in Proposition \ref{mem}. 

\item 
A unit scheme which can be broken into blocks of $8$ in an $8n\ti 8n$ enables for example convolutional codes with memory $7$ and rate $\f{1}{8}$ and $\f{7}{8}$  to be established. In special cases the $\f{7}{8}$ rate code is dual containing establishing a rate $\f{3}{4}$ quantum convolutional code of memory $8$ similar to Proposition \ref{mem}. 
\item 
  These types of constructions may be continued. For instance matrices with blocks of size $n$ and matrix size  $2^in \ti 2^in$ are more amenable; when the  matrix is orthogonal then dual-containing convolutional codes are obtainable from which quantum error correcting convolutional codes are formed with rate $\f{2^{i-1}-1}{2^{i-1}}$. Details are omitted.
  \end{enumerate}
      \subsection{Hadamard type unit codes}\label{hadamard} General constructions for linear block and convolutional codes  using the unit structure of Hadamard matrices are developed separately in \cite{induced}. 
      Hadamard matrices are nice orthogonal type units and the unit-derived structure can  be used to construct linear block and convolutional codes and to required types such as LDC, self-dual, DC and quantum codes. The unit-derived codes from general Hadamard matrices  are {\em not} the usual type of Walsh-Hadamard codes as known which have small rates,   but these can have nice features, required  rates, great relative distances and can be constructed to large rates.  The general process is developed in \cite{induced} and large lengths, types, distances, rates  are achievable. 
        
The Walsh-Hadamard codes of the known type $[2^s,s,s^{s-1}],[2^s,s+1,s^{s-1}]$ (over $GF(2)$) can be shown to be unit-derived codes formed from Walsh-Hadamard $2^s\ti 2^s$ matrices in a certain way -- see \cite{induced}. 
      For unit-derived codes, any Hadamard matrix may be used with which to construct unit-derived linear block and convolutional codes, and  DC, self-dual, LCD and quantum codes can be  produced therefrom. 
The distances can often be directly calculated from the Hadamard matrix used in the construction. 
          
      A prototype example is given below demonstrating how the process is derived from Hadamard matrices thus  illustrating the extent of these constructions from general Hadamard matrices.  
          
      \begin{example}\label{hadaex} Let $H$ be a Hadamard $12 \ti 12$ matrix. Here the computer algebra system GAP \cite{gap} is used to generate $H$ and the subsystem Guava is used to construct the codes and verify their distances in the linear block cases. The distances for the convolutional codes can be determined algebraically from the distances of the associated linear block codes.

        Thus $H$ has the unit form $\begin{ssmatrix}A \\ B\end{ssmatrix}\begin{ssmatrix}A\T & B\T \end{ssmatrix}= 12I_{12}$.
                In any field not of characteristic $2$ or $3$, $A$ and $B$ generate LCD codes. 
                \begin{itemize} \item
                  Three rows of $H$ generate a $[12,3,6]$ LCD code. \item six rows generate a $[12,6,6]$ LCD code. \item Nine rows generate a $[12,9,2]$ LCD code. 

                \item   Let $A$ be the first six rows of $H$ and $B$ be the last six rows of $H$. Define $G(z) = A+Bz$. Then $G(z)$ generates a $(12,6,6;1,12)$ convolutional code. \item Define $G(z) = A+iBz$ where $i$ is a square root of $(-1)$ in the field or in a quadratic extension of the field. Then $G(z)$ generates a self-dual convolutional $(12,6,6;1,12)$. From this a convolutional quantum code of type $[[12,0,12]]$ is formed.
                  
                    $GF(5)$ has $2$ as a square root of $(-1)$ so over $GF(5)$, $i$ can be taken to be $2$. Arithmetic in $GF(5)$ is modular arithmetic. $GF(7)$ does not contain a square root of $(-1)$ so it needs to be extended to $GF(7^2)$ in which to obtain a self-dual convolutional code.  

                  \item  Dual-containing convolutional codes of form $(12,9,3;1,d)$ are obtained by letting $A$ be the first nine rows of $H$ and $B$ the last three rows of $H$ and defining $G(z) = A+iB_1z$ where $B_1$ has first six rows consisting of zeros and last three rows consist of $B$.  This gives rise to a quantum convolutional code of the form $[[12,6,d]]$.  The distance $d = 4$ but note that $P(z)G(z)$ has distance $\geq 5+(s-2)$ where $s$ is the support of the information vector $P(z)$.

                  \item Form $(I_{12},H)$. This is a $[24,12,8]$ code. Form $(I_{12},iH)$, where $i$ is a square of $(-1)$ in the field or in a quadratic extension of the field. This is a self-dual $[24,12,8]$ code. 

                    In $GF(5)$ the element $2$ is a square root of $(-1)$ and thus  $(I_{12},2X)$ gives a self-dual $[24,12,8]$ code in $GF(5)$. The field $GF(7)$ needs to be extended to $GF(7^2)$ and then a self-dual code is obtained.
                \end{itemize}
            \end{example}
\section{Low Density Parity Check Codes}\label{ldpc}
A low density parity check, LDPC,  code 
is one where  the check/control matrix of the code has a small number of non-zero entries compared to its length, see for example \cite{mackay}.

The methods devised in previous sections 
for constructing  linear block and convolutional codes are now used to construct LDPC linear and convolutional codes. What  is required is the scheme produces a check/control matrix with low density compared to its length. It is known that for best performance of LDPC codes, there should be be no short cycles in the control matrix and this  can be achieved by the methods. 

Given a unit scheme $UV=I$ 
unit-derived codes are formed by taking any $r$ rows of $U$ as generator matrix and a check matrix is obtained by eliminating the corresponding columns of $V$. Thus if $V$ itself is of low density  then any such code formed is an LDPC code; if in addition $V$ itself has no short cycles then any such code formed is an LDPC code with no short cycles.    

Thus given a unit scheme $UV=I$, where $V$ is of low density and has no short cycles, 
 choose any $r$ rows of $U$ to form the generator matrix of an $[n,r]$ code  $\C$ and deleting the corresponding $r$ columns in $V$ gives a check matrix for the code and the code $\C$ formed is  an LDPC code with no short cycles in the check or control matrix. The code can also be specified by choosing the columns of the low density matrix $V$ to form the control matrix and going to $U$ to choose the rows which form a generator matrix. 
 
This is done in \cite{hurley33}  for linear block codes. Therein  methods are derived using units in group rings to produce linear LDPC codes {\em and} to produce such LDPC codes with no short cycles in the the check matrices. Thus a unit system, $uv=1$, is constructed in a group ring where one of the elements $u,v$, say $v$,  has small {\em support} as a group ring element. The $uv=1$ is then mapped  to the corresponding matrix equation, $UV= I$, by a process given in\cite{hur3}, in which $V$ has low density. Then using unit-derived codes leads to the construction of LDPC codes as required and when $V$ has no short cycles it leads to the construction of LDPC codes with no short cycles in the check matrix. It can be ensured that $V$ has no short cycles by a condition, see \cite{hurley33},  on the group elements with non-zero coefficient used in forming the group ring element $v$ short cycles.

In that paper \cite{hurley33} simulations are made and the examples given are shown to outperform substantially  previously constructed ones of the same size and rate. Randomly selected  LDPC codes with no short cycles are produced from the same unit. The codes produced are of particular use  in applications and industry where low storage and low power may be  a requirement or necessary for better functioning.  

    Note that $U$, from which the generator matrix is derived, does not, necessarily, have low density which is good from the point of  the minimum distance of the code; however as stated in Mackay \cite{mackay}, ``Distance isn't everything''. 

    Thus using group rings, systems are constructed $UV=I_n$ in which $V$ has low density with no short cycles anywhere. This gives an enormous freedom in which to construct LDPC codes with no short cycles. Indeed eliminating  any $(n-r)$ columns of $V$ gives a control matrix, and the generator matrix is formed by using the rows from $U$ corresponding to the eliminated columns of $V$; the result is  an $[n,r]$ LDPC code. Thus given $UV=I_n$ where $V$ has low density and no short cycles allows the construction of {\em many} LDPC codes with no short cycles. 
 
    
    In previous sections, methods are given for constructing convolutional codes from the unit-derived formula $\begin{ssmatrix} A \\ B \end{ssmatrix}\begin{ssmatrix} C & D \end{ssmatrix} = I_n$ from $UV=I_n$ using all the components $A,B,C,D$.  Convolutional codes of higher memory are obtained by further breaking up the unit system in blocks. These techniques may also be applied for constructing convolutional LDPC codes with no short cycles in the control matrix.  

Considering the basic formula and  when $A,B$ have the same size, and $n$ is even, then $G(z) = A+Bz$ generates  a non-catastrophic convolutional code. 
The control matrix is $D-Cz$ and the dual code is generated by $C\T-D\T z$. If $V$ is of low density and has no short cycles then $C\T-D\T z$ is of low density and has no short cycles. Thus the codes derived is a convolutional LDPC code with no short cycles.


    It is difficult to describe explicitly the LDPC codes derived as for  applications large lengths are required. Note that the method is very general and length and rate achieved can be decided in advance. We will concentrate on extending two of the examples in \cite{hurley33} to construct LDC convolutional codes with no short cycles.

    Only basic information on group rings is required. A good nice  book on  group rings is \cite{seh} and also the basic information may be found online by a simple search.

    Low density convolution codes and with no short in the control matrix are constructed by applying the methods in the previous sections together with the methods described in \cite{hurley33} for constructing LDPC linear codes with no short cycles. The following algorithm describes the constructions in general:
\begin{alg}\label{ldpcalg}
    \begin{enumerate} \item In a group ring  with group size $n$ find a unit and its inverse $uv=1$ where $v$ has small support and no short cycles. The size $n$ of the group should be large and the support of $v$ relatively small compared to $n$.
    \item\label{sd} From $uv=1$ go over to the matrix embedding of the group ring in a ring of matrices of size $n\ti n$, as in \cite{hur3}, to get a unit scheme
      $UV=I_n$ of matrices where $V$ is of low density and has no short cycles. \item Choose $r$ columns of $V$ to eliminate to form an $n\ti (n-r)$ matrix which will be a control matrix for a $[n,r]$ code. A generator matrix for this  code is the $r\ti n$ matrix formed by selecting the $r$ rows from $U$ corresponding in order to the  $r$ columns eliminated from $V$.
    \item\label{sd1} The unit scheme from item \ref{sd} may be presented as $\begin{ssmatrix}A \\ B \end{ssmatrix} \begin{ssmatrix} C & D \end{ssmatrix}$ where $A$ has size $r\ti n$, $B$ has size $(n-r) \ti n$, $C$ has size $n\ti r$ and $D$ has size $n \ti (n-r)$. Now here both $C,D$ are of low density and have no short cycles. An  LDPC code with no short cycles in the control matrix is given by $AD=0$ as in \cite{hurley33}. But also notice $BC =0$ gives a LDPC code in addition.
    \item The unit scheme in item \ref{sd1} as in previous sections is extended to $G(z) = A + Bz$ when $A,B$ have the same size (in which case the rate is $1/2$) or, when $A$ has size heater than the size of $B$, to $G(z) = A + B_1z$ where $B_1$ is obtained by extending $B$ with zero rows to be the size of $A$. Then $G(z)$ generates a convolutional memory $1$ code  which is non-catastrophic and has low density control matrix with no short cycles. The control matrix is $D - C_1z$ where $C_1$ is $C$ or a submatrix of $C$ as explained above.
    \item Obtaining memory greater than $1$ from the unit matrix scheme $UV$ derived from the group ring unit $uv$ also follows in a similar manner as described earlier. Examples of such are given below.
    \end{enumerate}
    \end{alg}
Examples must be of large length in order to satisfy the low density criterion. 
    In general  the examples  in \cite{hurley33} are taken from unit-derived codes within 
$\Z_2(C_n\ti C_4)$, where $\Z_2 = GF(2)$ is  the field of two elements. 

``The matrices derived are then submatrices of circulant-by-circulant
matrices and are easy to program. They are not circulant and thus are
not cyclic codes. Having
circulant-by-circulant rather than circulant allows a natural spreading 
of the non-zero coefficients and gives better distance and better
performance.''  

Assume that $C_n$ is generated by $g$ and  $C_4$ is
generated by $h$. 
Every element in the group ring is then of the form:
$\di\sum_{i=0}^{n-1}(\al_ig^i + h\be_ig^i + h^2\gamma_ig^i +
h^3\delta_ig^i)$,
with $\al_i,\be_i,\gamma_i,\delta_i \in \Z_2$.

\begin{example} Examples of $[96,48]$ LDPC codes are given in sections 3.2, 3.4 of \cite{hurley33}. 
 These examples are derived from the group ring  $\Z_2(C_{24}\ti C_4)$. 

The check element
$v = g^{24-9} + g^{24 - 15} +g^{24-19} + hg^{24-3} + hg^{24-20}+
   h^2g^{24-22} + h^3g^{24-22} + h^3g^{24-12}$ is used to define an  LDPC linear code. 

Then  $v$ has  no short cycles in its matrix $V$ and just $8$  or less non-zero elements in each row and column. Any choice of columns of $V$ will give an LDPC block linear code.
A pattern to delete half the columns from the matrix $V$ of $v$ 
is  chosen  to produce a rate $1/2$ code and is simulated and compared to other LDPC codes, outperforming these even when random columns are chosen. 

This selection is then presented as $\begin{ssmatrix}A \\ B \end{ssmatrix} \begin{ssmatrix} C & D \end{ssmatrix}$ where $A$ has size $48\ti 96$, $B$ has size $48\ti 96$, $C$ has size $96\ti 48$ and $D$ has size $96 \ti 48$.
Then define $G(z)=A + Bz$ to obtain a convolutional $[96,48,48;1]$ low density parity check code. The control matrix is $D- Cz$ which is $D+Cz$ as the characteristic is $2$.

In this case the matrix $U$ has the form $\begin{ssmatrix} A\\B\\C\\D\end{ssmatrix}$ and $V$ has the form $\begin{ssmatrix}E &F&G&H \end{ssmatrix}$ where $A,B,C,D$ have size $24\ti 96$ and $E,F,G,H$ have size $96\ti 24$ and each of $E,F,G,H$ have low density. Then as in Proposition $G(z) = A+Bz+Cz^2+Dz^3$ defines a convolutional $(96,24, 24*3;1)$ convolutional code which has low density check matrix.  
  \end{example}
      


\end{document}